\documentclass[lettersize,journal]{IEEEtran}
\usepackage{amsmath,amsfonts}
\usepackage{algorithmic}
\usepackage{algorithm}
\usepackage{array}
\usepackage[caption=false,font=normalsize,labelfont=sf,textfont=sf]{subfig}
\usepackage{textcomp}
\usepackage{stfloats}
\usepackage{url}
\usepackage{verbatim}
\usepackage{graphicx}
\usepackage{cite}
\usepackage[american]{circuitikz}
\usepackage{amsthm}
\newtheorem{theorem}{Theorem}
\newtheorem{lemma}{Lemma}

\usepackage{mathrsfs}

\usepackage[nolist]{acronym}
\begin{acronym}
    \acro{KVL}{Kirchhoff voltage law}
    \acro{SOC}{second-order-cone}
    \acro{SOCP}{second-order-cone programming}
    \acro{QCQP}{quadratically constrained quadratic programming}
    \acro{PF}{power flow}
    \acro{OPF}{optimal power flow}
    \acro{AC}{alternating current}
    \acro{SDP}{semidefinite programming}
\end{acronym}

\begin{document}

\title{On the Tightness of the Second-Order Cone Relaxation of the Optimal Power Flow with Angles Recovery in Meshed Networks}

\author{Ginevra Larroux, Matthieu Jacobs, Mario Paolone~\IEEEmembership{Fellow,~IEEE}
\thanks{The authors are with the École Polytechnique Fédérale de Lausanne (EPFL), Switzerland, email: \{ginevra.larroux,mario.paolone\}@epfl.ch.}}

\maketitle

\begin{abstract}
This letter investigates properties of the second-order cone relaxation of the optimal power flow (OPF) problem, with emphasis on relaxation tightness, nodal voltage angles recovery, and alternating-current-OPF feasibility in meshed networks. The theoretical discussion is supported by numerical experiments on standard IEEE test cases. Implications for power system planning are briefly outlined.
\end{abstract}


\vspace{-5mm}
\section{Introduction}\label{introduction}
Convex relaxations of the \ac{AC}-\ac{OPF} problem are widely used in power system optimization, both as computationally tractable surrogates and as tools to assess feasibility and optimality. Among these, \ac{SOC} relaxations based on branch-flow grid models have attracted particular attention because they can be solved efficiently with standard off-the-shelf \ac{SOCP}/\ac{QCQP} solvers (e.g., interior-point methods). Milestones in the development of \ac{SOC}-\ac{OPF} formulations include, in chronological order, \cite{Wu_1, Jabr_2, Farivar_3, Low_PartI, Low_PartII, Zhao}. 

A central question for any convex relaxation is \emph{tightness}, i.e., whether a solution of the relaxed problem corresponds to a feasible solution of the original problem. While \ac{SOC} branch-flow relaxations can be exact under known conditions in radial networks, cycle-induced consistency constraints on voltage angles make angle recovery and tightness nontrivial in meshed networks.

This letter focuses on two widely used \ac{SOC}-\ac{OPF} formulations: the baseline model in \cite{Low_PartI}, which eliminates bus voltage and branch current angles and relaxes the nonconvex power-loss equalities, with angle recovery considered only \emph{a posteriori}; and the formulation in \cite{Zhao}, which instead retains bus voltage angles explicitly in the network constraints. The main thesis of this letter is that, for meshed networks, the model in \cite{Zhao} is not a relaxation of the \ac{AC}-\ac{OPF} but rather an approximation; consequently, a zero relaxation gap does not, in general, guarantee \ac{AC}-\ac{OPF} feasibility in meshed networks.

The remainder of the paper is organized as follows. Section~\ref{sec:problem_formulation} introduces the two formulations and their properties. Section~\ref{sec:discussion} develops a theoretical argument and provides a counterexample to Theorem~1 in \cite{Zhao}. Section~\ref{sec:conclusions} discusses implications for planning applications and concludes the paper.

\section{Problem formulation}\label{sec:problem_formulation}
\subsection{\ac{SOC}-\ac{OPF} formulation by \cite{Low_PartI}}\label{subsec:low_partI}
For brevity, we refer to \cite{Low_PartI} for the full nomenclature; unless stated otherwise, constraints are understood to hold for all buses in $N$ and branches in $E$.
Given the power injections $s_i:=s_i^g - s_i^c$ for $i=1,\dots,n$, the \ac{PF} variables are
$x(s):=(S,I,V,s_0)$, where $S:=\{S_{ij}\}_{(i,j)\in E}$, $I:=\{I_{ij}\}_{(i,j)\in E}$, $V:=\{V_i\}_{i\in N\setminus\{0\}}$, and $s_0$ denotes the slack injection.
The \ac{AC}-\ac{PF} equations are listed in \eqref{eq:Low_PF}. 
\begin{subequations}\label{eq:Low_PF} 
    \begin{alignat}{4}
        & s_j =  \sum_{k: j \rightarrow k} S_{j k}-\sum_{i: i \rightarrow j}\left(S_{i j}-z_{i j} \lvert I_{i j}\rvert^2\right) +y_j^* \lvert V_j\rvert^2\\
        & V_i - V_j = z_{ij} I_{ij} \\
        & S_{ij} = V_iI_{ij}^* 
    \end{alignat}
\end{subequations}

In the associated \ac{OPF}, the injections $s$ are also decision variables and the objective is chosen convex and independent of phase angles, e.g.,
$f(\hat h(x),s):=\sum_{(i,j)\in E} r_{ij} |I_{ij}|^2$.
We refer to \cite{Low_PartI} for the definition of the projection $\hat h(\cdot)$.
Operational constraints include bounds on generations and loads, voltage magnitudes, and branch currents or powers. For brevity, we explicitly report only the ampacity constraint and use \eqref{eq:Low_operational_cons} to refer to \emph{all} operational constraints.
\begin{equation}\label{eq:Low_operational_cons} 
        \lvert I_{ij}\rvert \leq \overline{I}_{ij}
\end{equation}
The resulting non-approximated \ac{OPF} problem is
\begin{equation}\label{eq:Low_OPF_rewrite}
(\mathrm{OPF})\quad \min_{x,s} \ f(\hat h(x),s)\qquad\text{s.t.}\quad \eqref{eq:Low_PF},\eqref{eq:Low_operational_cons}.
\end{equation}

The relaxation in \cite{Low_PartI} is obtained in two stages. First, voltage and current angles are eliminated by replacing \eqref{eq:Low_PF} with \eqref{eq:Low_OPFar}, yielding the (OPF-ar) model, in the lifted variables $l_{ij}:=|I_{ij}|^2$ and $v_i:=|V_i|^2$. Second, the nonconvex quadratic equality \eqref{eq:Low_q_ol_equality} is relaxed into a second-order cone inequality, resulting in the convex model (OPF-cr).
\begin{subequations}\label{eq:Low_OPFar} 
    \begin{alignat}{4}
        & p_j=\sum_{k: j \rightarrow k} P_{j k}-\sum_{i: i \rightarrow j}\left(P_{i j}-r_{i j} \ell_{i j}\right)+g_j v_j\\
        & q_j=\sum_{k: j \rightarrow k} Q_{j k}-\sum_{i: i \rightarrow j}\left(Q_{i j}-x_{i j} \ell_{i j}\right)+b_j v_j\\
        & v_j=v_i-2\left(r_{i j} P_{i j}+x_{i j} Q_{i j}\right)+\left(r_{i j}^2+x_{i j}^2\right) \ell_{i j} \\
        & l_{ij}=\frac{P_{i j}^2+Q_{i j}^2}{v_i} \label{eq:Low_q_ol_equality} 
    \end{alignat}
\end{subequations}

If the hypotheses\footnote{(I) The network graph $G$ is connected; (II) the cost function is convex; (III) the cost function is strictly increasing in $l$, non-increasing in $s^c$, and independent of $S$; (IV) the optimal \ac{AC}-\ac{OPF} is feasible.} hold, the authors show that any optimal solution of (OPF-cr) is also optimal for (OPF-ar), even on meshed networks, provided there are no upper bounds on active and reactive loads. Crucially, they also note that recovering angles from an (OPF-ar) solution is always possible in radial networks, but not necessarily in meshed networks. Theorem~2 in \cite{Low_PartI} gives a condition to determine whether a branch-flow solution can be recovered from an (OPF-ar) solution, together with the corresponding recovery computation. As this is nonconvex, two angle-recovery algorithms (``centralized'' and ``distributed'') are proposed, but they are not guaranteed to succeed in meshed networks. In the related work \cite{Low_PartII}, a convexification approach for meshed networks based on phase shifters is proposed when the angle-recovery condition fails. 

On a final note, \cite{ExactConvexRelaxation_Radial} discusses alternative sufficient conditions for exactness in radial networks that do not rely on allowing load oversatisfaction (unrealistic in practice).

\subsection{\ac{SOC}-\ac{OPF} formulation by \cite{Zhao}}\label{subsec:zhao} 
Reference \cite{Zhao} proposes a formulation that explicitly includes bus voltage angles and distinguishes between \emph{measurable} branch quantities (denoted with $\tilde \cdot$) and branch-flow \ac{OPF} variables that exclude shunt contributions, enabling the correct treatment of branch current ampacity limits. 

Using the notation of \cite{Zhao}, the \ac{AC}-\ac{PF} equations are in \eqref{eq:Zhao_PF}.
\begin{subequations}\label{eq:Zhao_PF}
    \begin{alignat}{1}
        p_n-p_{d_n}=\sum_l\left(A_{nl}^{+} p_{s_l}-A_{nl}^{-} p_{o_l}\right)+G_n V_n \label{eq:Zhao_PF1} \\ 
        q_n-q_{d_n}=\sum_l\left(A_{nl}^{+} q_{s_l}-A_{nl}^{-} q_{o_l}\right)-B_n V_n \label{eq:Zhao_PF2}\\
        V_{s_l}-V_{r_l}=2 R_l p_{s_l}+2 X_l q_{s_l}-R_l p_{o_l}-X_l q_{o_l} \label{eq:Zhao_PF3}\\
        v_{s_l} v_{r_l} \sin \theta_l=X_l p_{s_l}-R_l q_{s_l}\label{eq:KVL_phase} \\
        V_n=v_n^2\\
        \theta_l= \theta_{s_l}-\theta_{r_l}\label{eq:theta_def}\\
        q_{o_l}=\frac{p_{s_l}^2+q_{s_l}^2}{V_{s l}} X_l\label{eq:reactive_pow_loss} \\
        p_{o_l} X_l=q_{o_l} R_l \label{eq:Zhao_PF5}
    \end{alignat} 
\end{subequations}
Define $A_{s_ln},A_{r_ln} \in \mathbb{R}^{|\mathscr{L}|\times|\mathscr{N}|}$ the branch-to-node incidence matrices with $A_{s(r)_ln}=1$ if node $n$ is the sending (receiving) end of branch $l$. Then \eqref{eq:theta_def} can be equivalently written as \eqref{eq:theta_def_incidence_matrices}.
\begin{equation}\label{eq:theta_def_incidence_matrices}
\theta_l = (A_{s_ln}-A_{r_ln}) \ \theta_n
\end{equation}

As above, for brevity we explicitly report only the ampacity constraint among the operational limits. Note that, in this formulation, bounds are also imposed on bus voltage angles and on branch angle differences.
\begin{equation}\label{eq:ampacity}
    \tilde{i}_{s\left(r\right)_l}^2=\frac{\tilde{p}_{s\left(r\right)_l}^2+\tilde{q}_{s\left(r\right)_l}^2}{V_{s\left(r\right)_l}} \leq \widetilde{K}_l 
\end{equation}

The corresponding non-approximated \ac{AC}-\ac{OPF} is written as
\begin{equation}\label{eq:Zhao_ACOPF}
(\text{o-ACOPF})\quad \min_{\Omega_{\text{o-ACOPF}}} \ f(\Omega_{\text{o-ACOPF}}) \qquad\text{s.t.} \quad \eqref{eq:Zhao_PF},\eqref{eq:ampacity}
\end{equation}
where $f(\cdot)$ is a convex objective function.

The proposed \ac{SOC}-\ac{OPF} formulation then modifies two constraints: the non-convex relation \eqref{eq:KVL_phase} is replaced by the linear \eqref{eq:soc_eq2d}, and the losses equality~\eqref{eq:reactive_pow_loss} is relaxed to the conic inequality~\eqref{eq:soc_eq2b}. In addition, the ampacity constraint~\eqref{eq:ampacity} is rewritten in terms of measurable quantities \eqref{eq:ampacity_meas}. Finally, a tightness-promoting constraint \eqref{eq:recover_feasibility} is introduced under the hypothesis
$(\theta_l^{\min}, \theta_l^{\max}) \subseteq \left(-\frac{\pi}{2},\frac{\pi}{2}\right)$, $\theta_l^{\min} = -\theta_l^{\max}$.
\begin{subequations}\label{eq:Zhao_SOC}
    \begin{alignat}{1}
        & \theta_l=X_l p_{s_l}-R_l q_{s_l} \label{eq:soc_eq2d} \\
        & q_{o_l}\geq \frac{p_{s(r)_l}^2+q_{s(r) l}^2}{V_{s(r)_l}} X_l \label{eq:soc_eq2b} \\
        & K_{o_l} =\left(\widetilde{K}_l-V_{s(r)_l} B_{s(r)_l}^2+2 q_{s(r)_l} B_{s(r)_l}\right) X_l \geq q_{o_ l}\label{eq:ampacity_meas} \\
        & V_{s_l} V_{r_l} \sin ^2\left(\theta_l^{\max }\right) \geq \theta_{l}^2 \label{eq:recover_feasibility}
    \end{alignat} 
\end{subequations}
The resulting \ac{SOC}-\ac{OPF} model is
\begin{equation}\label{eq:Zhao_SOCOPF}
\begin{aligned}
    (\text{SOC-ACOPF})\quad &\min_{\Omega_{\text{SOC-ACOPF}}} \ f(\Omega_{\text{SOC-ACOPF}}) \\
    &\text{s.t.}\quad \eqref{eq:Zhao_PF1}, \eqref{eq:Zhao_PF2}, \eqref{eq:Zhao_PF3}, \eqref{eq:theta_def}, \eqref{eq:Zhao_PF5}, \eqref{eq:Zhao_SOC}.
\end{aligned}
\end{equation}

Reference \cite{Zhao} provides several theoretical results on tightness and on the recovery of \ac{AC}-\ac{OPF}-feasible operating points. The theorems most relevant to this letter are recalled below.
\begin{theorem}[Theorem 1 in \cite{Zhao}]\label{th:theo1}
    Assume $(\theta_l^{\min}, \theta_l^{\max}) = \left(-\frac{\pi}{2},\frac{\pi}{2}\right)$ and $(v_n^{\min}, v_n^{\max})=(0.9,1.1)$, replacing \eqref{eq:KVL_phase} with \eqref{eq:soc_eq2d} relaxes the (o-ACOPF) problem. 
\end{theorem}
\begin{theorem}[Theorem 4 in \cite{Zhao}]\label{th:theo4}
    If the optimal solution $\Omega^*$ of the (SOC-ACOPF) model with the additional constraint \eqref{eq:recover_feasibility} gives a tight solution $q_{o_l}^*=\frac{p_{s_l}^{*2}+q_{s_l}^{*2}}{V_{s_l}^{*}}X_l$, $\forall l \in \mathcal{L}$, the global optimum solution $\Omega_0^*$ of the (o-ACOPF) model is $\Omega_0^*:=\{\Omega^* \setminus (\theta_l^*, V_n^*)\} \cup \left\{ v_{0,n}^*:= \sqrt{V_n^*}, \ \theta_{0,l}^* :=\arcsin{\left(\frac{\theta_l^*}{v_{s_l}^*v_{r_l}^*}\right)}\right\}$.
\end{theorem}

\section{Discussion}\label{sec:discussion}
This section provides a counterargument to Theorem~1 in \cite{Zhao}. We show that there exist operating points that are feasible for the \ac{AC}-\ac{PF} (and hence belong to the feasible region of the corresponding \ac{AC}-\ac{OPF}, the (o-ACOPF) model in \cite{Zhao}), but that do not belong to the feasible set of the \ac{SOC}-\ac{OPF} model proposed in \cite{Zhao} (SOC-ACOPF) once the linearized angle relation \eqref{eq:soc_eq2d} is imposed together with angle consistency in meshed networks. Consequently, in meshed networks, (SOC-ACOPF) is an \emph{approximation} rather than a \emph{relaxation} of (o-ACOPF). This also undermines Theorems~2--4 in \cite{Zhao}, whose proofs build upon Theorem~1.
\vspace{-2mm}
\subsection{Theoretical considerations}
Constraint \eqref{eq:soc_eq2d} is obtained by linearizing the nonlinear relation \eqref{eq:KVL_phase} around the operating point $\theta_l \approx 0$ and $v_n \approx 1$. The proof strategy in Theorem~1 of \cite{Zhao} argues that, under the bounds $(\theta_l^{\min},\theta_l^{\max}) = \left(-\frac{\pi}{2},\frac{\pi}{2}\right)$ and $(v_n^{\min},v_n^{\max})=(0.9,1.1)$, the feasible set induced by the linear constraint \eqref{eq:soc_eq2d} contains that of the nonlinear term $v_{s_l} v_{r_l}\sin(\theta_l)$, hence suggesting that the modified model is a relaxation. However, this argument does not account for cycle consistency of voltage angles \eqref{eq:cyclic_con} in meshed networks.
\begin{equation}\label{eq:cyclic_con}
    \sum_{l\in\mathscr{C}} \theta_l = 0 \bmod 2\pi, \qquad \forall \text{ cycles } \mathscr{C} 
\end{equation}

In \cite{Zhao} it is stated that, introducing $\theta_n$ as decision variables and defining $\theta_l$ through \eqref{eq:theta_def}, ``implicitly'' enforces the cyclic condition \eqref{eq:cyclic_con}. Lemma~\ref{lem:angle_incompatibility} highlights where the relaxation claim fails for meshed grids.

\begin{lemma}\label{lem:angle_incompatibility}
Consider a meshed network and nodal voltage angle differences across branches defined by \eqref{eq:theta_def_incidence_matrices}. If the non-approximated branch \ac{KVL} equation \eqref{eq:KVL_phase} is replaced by the linearized constraint \eqref{eq:soc_eq2d} in the power flow equations, then the resulting branch angle differences are not guaranteed to satisfy the cycle constraint \eqref{eq:cyclic_con}.
\end{lemma}
\begin{proof}
In radial networks, for \emph{any} $\theta_l$, the nodal angles $\theta_n$ are determined uniquely up to the slack reference, since $A_{s_ln}-A_{r_ln}$ is full rank and the Rouché--Capelli theorem holds by construction. In meshed networks, by contrast, $\theta_l$ must satisfy the additional cycle constraints \eqref{eq:cyclic_con}. Even when the conic inequality \eqref{eq:soc_eq2b} is tight (i.e., the relaxation gap is zero), if the (overdetermined) system $\left(A_{s_ln}-A_{r_ln},\,\theta_l\right)$ satisfies a solution $\theta_n$, the system $\left(A_{s_ln}-A_{r_ln},\, \arcsin{\left(\frac{\theta_l}{(A_{s_ln}-A_{r_ln}) \sqrt{V_n}}\right)}\right)$
may not be satisfied, unless \eqref{eq:soc_eq2d} and \eqref{eq:KVL_phase} are simultaneously satisfied. This can occur only in degenerate cases, i.e., when $v_{s_l}v_{r_l}=1$ and $\theta_l\approx 0$.
\end{proof}

\vspace{-2mm}
\subsection{Numerical counterexample}
We consider two standard IEEE test systems in \texttt{pandapower}: the IEEE 33-bus system (radial) and the IEEE 39-bus system (meshed). For each case, we start from the \texttt{pandapower} \ac{AC}-\ac{PF} solution, denoted by $(p_{n}^0, q_{n}^0, v_{n}^0, \theta_{n}^0)$.

\paragraph{Step 1: verify feasibility in the non-relaxed (o-ACOPF) model}
We substitute $(p_{n}^0, q_{n}^0, v_{n}^0, \theta_{n}^0)$ into the non-relaxed formulation of \cite{Zhao} in \eqref{eq:Zhao_PF} by rewriting it as the linear system $A x = b$ in \eqref{eq:Axb}.
\begin{equation}\label{eq:Axb}
{\footnotesize
    \begin{bmatrix}
        A_{nl}^+ & Z_n & -A_{nl}^- & Z_n \\
        Z_n & A_{nl}^+ & Z_n & -A_{nl}^- \\
        2R & 2X & -R & -X \\
        X & -R & Z_l & Z_l \\
        Z_l & Z_l & X & -R 
    \end{bmatrix}
    \begin{bmatrix}
        p_{s_l} \\ q_{s_l} \\ p_{o_l} \\ q_{o_l}
    \end{bmatrix}
    =
    \begin{bmatrix}
        p_{n}^0 -  G_n V_{n}^0 \\
        q_{n}^0 + B_n V_{n}^0 \\
        A_l V_{n}^0 \\
        A_l v_{n}^0 \sin(\theta_{l}^0) \\
        \mathbf{0}_{|\mathscr{L}|}
    \end{bmatrix}}
\end{equation}
Here $Z_n := \mathbf{0}_{|\mathscr{N}|\times|\mathscr{L}|}$ and $Z_l := \mathbf{0}_{|\mathscr{L}|\times|\mathscr{L}|}$ denote zero matrices, $R:=\operatorname{diag}(R_l)$, $X:=\operatorname{diag}(X_l)$, and $A_l := A_{s_ln} - A_{r_ln}$. The branch angle differences are obtained from bus angles using their definition, $\theta_{l}^0 = A_l \theta_{n}^0$. We solve \eqref{eq:Axb} numerically using a least-squares method and verify that the residuals are (up to numerical tolerances) zero.

\paragraph{Step 2: test cycle consistency of the voltage angle differences across branches implied by the (SOC-ACOPF) model}
We now construct the voltage angle differences across branches under the (SOC-ACOPF) model, i.e. $\theta_{l}^1 := v_{s_l}^0 v_{r_l}^0\sin(\theta_{l}^0)$.
We then test whether it is compatible with a set of bus voltage angles by checking if there exists $\theta_{n}^1$ such that
$(A_{sl}-A_{rl})\theta_{n}^1 = \theta_{l}^1$. We solve this system in least squares after removing the slack-bus column (gauge invariance), yielding an exact square system for the radial case and an overdetermined system for the meshed network. The residual norm is zero if and only if the cycle constraints are satisfied.

Residual distributions for both networks are reported in Fig.~\ref{fig:radial_vs_meshed_boxplot}. For the system $(A,b)$, residuals are numerically zero in both cases, confirming (o-ACOPF) feasibility of the \texttt{pandapower} operating points. For the system $(C,d)$ with $C=A_{s_ln}-A_{r_ln}$ and $d=\theta_{l}^1$, residuals are numerically zero for the radial network (as expected), whereas they are non-negligible for the meshed network, with a maximum residual of $0.0245$ degrees. This deviation, when propagated through the \ac{AC}-\ac{PF} equations, yields branch power-flow errors on the order of $0.5$~p.u., indicating that $\theta_{l}^1$ violates cycle consistency.

\begin{figure}
    \centering
    \includegraphics[width=0.6\linewidth]{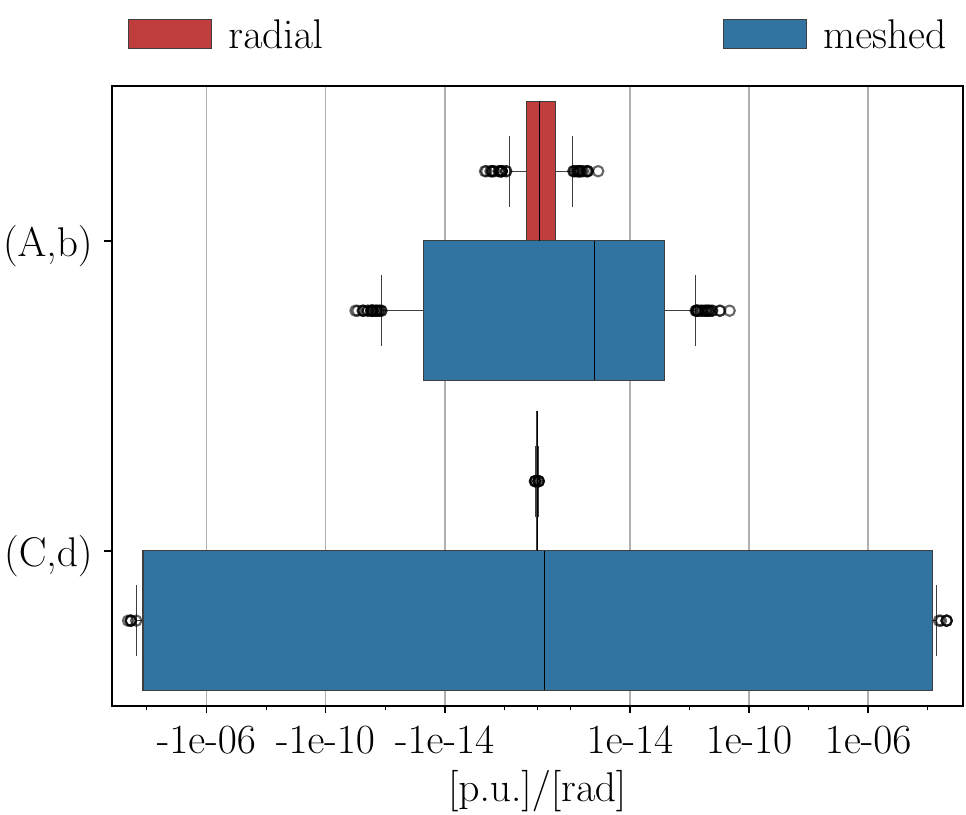}
    \caption{Residuals when solving the linear systems $(A x \approx b)$ and $(C\theta_n \approx d)$, for the radial IEEE 33-bus (red) and meshed IEEE 39-bus (blue) networks.}
    \label{fig:radial_vs_meshed_boxplot}\vspace{-5mm}
\end{figure}

\section{Conclusions}\label{sec:conclusions}
This letter shows that, in meshed networks, the \ac{SOC}-\ac{OPF} formulation in \cite{Zhao} cannot guarantee \ac{AC}-\ac{OPF} feasibility, even when the conic relaxation of the quadratic loss equality is tight.
At the same time, \cite{Zhao} provides a practically relevant way to incorporate voltage angles directly into the \ac{OPF} constraints. Although this yields an approximation of the \ac{AC}-\ac{OPF}, the resulting operating point is often close to an \ac{AC}-\ac{PF}-feasible one. 

By contrast, formulations such as \cite{Low_PartI} rely on a posteriori cycle-consistency checks and angle recovery, which may fail and thus yield a solution that is infeasible or arbitrarily inaccurate. The approach in \cite{Low_PartII} instead restores exactness in meshed networks via phase shifters, but at a cost that may be of questionable economic viability.

In both cases, an a posteriori procedure is needed to obtain an \ac{AC}-\ac{OPF} feasible solution even when the relaxation is tight.

\bibliographystyle{IEEEtran}
\bibliography{references.bib}

\end{document}